\documentclass[10pt,journal,final,finalsubmission,twocolumn]{IEEEtran}
\makeatletter

\let\proof\@undefined
\let\endproof\@undefined
\makeatother

\usepackage{amssymb}
\usepackage{amsmath}
\usepackage{amsthm}
\usepackage{amsfonts}
\usepackage{algorithm}
\usepackage{algorithmic}

\usepackage[pdftex]{color}

\theoremstyle{plain}

\newtheorem{proposition}{Proposition}

\numberwithin{equation}{section}

{\begin{list}{}%
         {\setlength{\leftmargin}{#1}}%
         \item[]%
}
{\end{list}}

\usepackage{graphicx}  

\def\F{\mathcal{F}}
\def\G{\mathcal{G}}

\definecolor{orange}{rgb}{1,0.5,0}

\begin{document}
%
\title{Nearness to Local Subspace Algorithm for Subspace and Motion Segmentation}
%
%
\author{Akram~Aldroubi
        and~Ali~Sekmen,~\IEEEmembership{Member,~IEEE}
\thanks{A. Akram and A. Sekmen are with the Department
Mathematics,Vanderbilt University, Nashville,
TN, 37212 USA e-mail: akram.aldroubi@vanderbilt.edu}
\thanks{A. Sekmen is also with the Department of Computer Science at Tennessee State 
Univesity, Nashville, TN 37209 USA email:asekmen@tnstate.edu }
\thanks{ The research of A.~Aldroubi is supported in part by NSF Grant DMS-0807464.}}
\maketitle

\begin{abstract}
There is a growing interest in computer science, engineering, and mathematics for modeling signals in terms of union of subspaces and manifolds. Subspace segmentation and clustering of high dimensional data drawn from a union of subspaces are especially important with many practical applications in computer vision, image and signal processing, communications, and information theory. This paper presents a clustering algorithm for high dimensional data that comes from a union of lower dimensional subspaces of equal and known dimensions. Such cases occur in many data clustering problems, such as motion segmentation and face recognition. The algorithm is reliable in the presence of noise, and applied to the Hopkins 155 Dataset, it generates the best results to  date for motion segmentation.  The two motion, three motion, and overall segmentaion rates for the video sequences are 99.43$\%$,  98.69$\%$, and  99.24$\%$, respectively.
\end{abstract}


\section{Introduction}
\label{section:introduction}

The problem of subspace clustering is to find a nonlinear model of the form $\mathcal{U}=\bigcup_{i\in I}S_i$ where $\left\{S_i\right\}_{i\in I}$ is a set of  subspaces that  is nearest to a set of data $\textbf{W}=\left\{w_1,...,w_N\right\} \in \mathbb{R}^d$.  The model can then be used to classify the data $\textbf{W}$ into classes called clusters.

In many engineering and mathematics applications, data lives in a union of low dimensional subspaces \cite{Lu08,Kanatani03,Akram09,Vidalbook}.  For instance, consider a moving affine camera that captures $F$ frames of a scene that contains multiple moving objects. Let $p$ be a point of one of these objects and let $x_i(p), y_i(p)$ be the coordinates of $p$ in frame $i$. Define the   {\em trajectory vector} of $p$ as the vector $w(p)=(x_1(p),y_1(p),x_2(p),y_2(p),\dots,x_N(p),y_N(p))^t$ in $ \mathbb R^{2F}$.  It can be shown that the trajectory vectors of all points of an object in a video belong to a vector subspace in $\mathbb R^{2F}$ of dimension no larger than $4$ \cite{Kanatani01,Akram10}. Thus, trajectory vectors in videos can be modeled by a union $\mathcal{M}=\cup_{i\in I}V_i$ of $l$ subspaces where $l$ is the number of moving objects (background is itself a motion). It can also be shown that human facial motion and other non-rigid motions can be approximated by linear subspaces \cite{Bregler00,Brand01}. Another clustering problem that can be modeled as union of subspaces is recognition of faces. Specifically, the set of all two dimensional images of a given face $i$, obtained under different illuminations and facial positions, can be modeled as a set of vectors belonging to a low dimensional subspace $S_i$ living in a higher dimensional space $\mathbb{R}^d$ \cite{Basri01,Ho03,Vidalbook}. A set of such images from different faces is then a union $\mathcal{U}=\bigcup_{i\in I}S_i$. Similar nonlinear models arise in sampling theory where $\mathbb{R}^d$ is replaced by an infinite dimensional Hilbert space $\mathcal{H}$, e.g., $L^2(\mathbb{R}^d)$ \cite{AT11,Akram08,Lu08,Maravic05}.

\subsection{Subspace Segmentation Problem}
\label {SSP}
The goal of subspace clustering is to identify all of the subspaces that a set of data $\textbf{W}=\left\{w_1,...,w_N\right\} \in \mathbb{R}^d$ is drawn from and assign each data point $w_i$ to the subspace it belongs to. The number of subspaces, their dimensions, and a basis for each subspace are to be determined.  The subspace clustering or segmentation problem can be stated as follows:
\begin{changemargin}{1cm}{1cm} 
Let $\mathcal{U}=\bigcup_{i=1}^{M}S_i$ where $\left\{S_i \subset \mathcal{H}\right\}_{i=1}^{M}$ is a set of  subspaces of a  Hilbert space $\mathcal{H}$. Let $\textbf{W}=\left\{w_j \in \mathcal{H}\right\}_{j=1}^{N}$ be a set of data points drawn from $\mathcal{U}$. Then,
\begin{changemargin}{0.1cm}{0.1cm} 
\begin{enumerate}
	\item determine the number of subspaces $M$,
	\item determine the set of dimensions $\left\{d_i\right\}_{i=1}^{M}$,
	\item find an orthonormal basis for each subspace $S_i$,
	\item collect the data points belonging to the same subspace into the same cluster. \\
\end{enumerate}
\end{changemargin}
\end{changemargin} 
Note that often the data may be corrupted by noise, may have outliers or the data may not be complete, e.g., there may be missing data points.  In some subspace clustering problems, the number $M$ of subspaces  or the dimensions of the subspaces $\{d_i\}_{i=1}^{M}$ are known. A number of  approaches have been devised to solve the problem above or some of its special cases.

\subsubsection{Sparsity Methods}
 Elhamifar \textit{et al.} developed an algorithm for linear and affine subspace clustering using sparse representation of vectors \cite{Elhamifar09,Elhamifar10}. This method combined with a spectral clustering, gives good results for motion segmentation and it is more general than  Eldar's work in compressed sensing \cite{Eldar09}.  Another method, related to compressed sensing  by Liu \textit{et al.} \cite{Liu10, Liu10_2} finds the lowest rank representation of the data matrix. The lowest rank representation is then  used to define the similarity of an undirected graph, which is then followed by spectral clustering. Favaro \textit{et al.} in \cite{Favaro11} extends \cite{Elhamifar09,Elhamifar10,Liu10,Liu10_2}.

\subsubsection{Algebraic Methods}
Algebraic methods have also been used for solving the subspace clustering problem. The 
Generalized Principle Component Analysis (GPCA) is one such method \cite{Vidalbook,Vidal05,Vidal07}, and it can distinguish subspaces of different dimensions. Since it is algebraic, it is computationally inexpensive, however, its complexity increases exponentially as the number of subspaces and their dimensions increase. It is also  sensitive to noise and outliers.  The Robust Algebraic Segmentation is a more specialized algebraic method developed by Rao \textit{et al.} \cite{Rao10} to partition image correspondences to the motions in a 3-D dynamic scene (that contains 3-D rigid body and 2-D planar structures) under perspective camera projection. 

\subsubsection{Iterative and Statistical Methods}
Iterative methods have also been employed for the subspace clustering problem. For example, the  nonlinear least squares \cite{Akram08,Akram09} and K-subspaces \cite{Tseng00} start with an initial estimation of subspaces (or estimation of the bases of the subspaces). Then, a cost function reflecting the ``distance'' of a point  to each subspace is computed and the point is assigned to its closest subspace. After that, each cluster of data is used to reestimate each subspace. The procedure is repeated until the segmentation of data points does not change. These methods, however, are sensitive to the initialization and require  a good initial partition for convergence to a global minimum.

The statistical methods such as Multi Stage Learning (MSL) \cite{Kanatani03,Gruber04} are typically based on Expectation Maximization (EM) \cite{Candillier05}.  The union of subspaces is modeled by a mixture of probability distributions. For example,  each subspace is modeled by a Gaussian distribution. The model parameters are then estimated using \textit{Maximum Likelihood Estimation}. This is done by using a two-step process  that optimizes the \textit{log-likelihood} of the model which depends on some hidden (latent) variables. In \textit{E-Step} (Expectation), the expectation of the \textit{log-likelihood} is computed using the current estimate of the latent variables. In \textit{M-Step} (Maximization), the values of the latent variables are updated by maximizing the expectation of the \textit{log-likelihood}. As in the case of the iterative methods, statistical methods highly depends on initialization of model parameters or segmentation and they  assume that the number of subspaces as well as their dimensions are known.

The Random Sample Consensus (RANSAC) \cite{Bolles81}, which has been applied to numerous computer vision problems, is successful in dealing with noise and outliers. But it is a specialized algorithm and assumes that the  subspaces have the same dimension and that this dimension is known.

\subsubsection{Spectral Clustering Methods}
\label{SpectralMethods}
Spectral clustering \cite{Luxburg07} is often used in conjunction with other methods as the final step in clustering.  Some of the latest subspace clustering algorithms (such as \cite{Elhamifar09,Elhamifar10,Lerman09_3}) aim at defining an appropriate similarity matrix  between data points which then can be used for further processing using the spectral clustering method. An application of spectral clustering to motion segmentation can be found in \cite{Schnorr09}. Spectral curvature clustering \cite{Lerman09,Lerman09_2} is a variant of spectral clustering. \cite{Zelnik04} provides a spectral clustering algorithm that aims at reducing the computational complexity.  The motion segmentation algorithm developed by Yan and Pollefeys \cite{Yan06} first estimates a local linear manifold for each trajectory data and then computes an affinity matrix based on the principle subspace angles between each pair of local linear manifolds. The algorithm then uses spectral clustering for segmenting the trajectories of independent, articulated, rigid, and non-rigid body motions. \cite{Vidal10} gives a detailed treatment of various related algorithms.


\subsection{Motion Segmentation Problem}
The appendix gives a detailed treatment of motion segmentation as a special case of the subspace segmentation problem. First, a data matrix $W_{2F\times N}$ is constructed using $N$ feature points that are tracked across $F$ frames. Then, each column of $W$ (i.e., the trajectory vector of a feature point) is treated as a data point and it is shown that all of the data points that correspond to the same moving object lie in an at most 4-dimensional subspace of $\mathbb{R}^{2F}$. 

\subsection{Paper  Contributions}
\begin{enumerate}
\item This paper presents a clustering algorithm for high dimensional data that are drawn from a union of low dimensional subspaces of equal and known dimensions. The algorithm is applicable to the motion segmentation problem and uses some fundamental linear algebra concepts. Some of our ideas are similar to those of Yan and Pollefeys described above  in Section \ref{SpectralMethods}. However, our algorithm differs from theirs fundamentally as described below:
\begin{itemize}
\item Yan and Pollefeys' method estimate a subspace $S_i$ for each point $x_i$, and then computes the principle angles between those subspaces as an affinity measure. In our work, we also estimate  a subspace for each point, however, these local subspaces are used differently. They are used to compute the distance between each point $x_j$ to the local subspace  $S_i$ for the data point $x_i$.
\item In their method, an exponential function for affinity of two points $x_i$ and $x_j$ is used, and this exponential function depends on the principle angles between the subspaces $S_i$ and $S_j$ that are associated with $x_i$ and $x_j$, respectively. In our case, the affinity measure is different. We first find the distance between $x_j$ and $S_i$ and then apply  a threshold, computed from the data,  to obtain a binary similarity matrix for all data points.
\item The method of Yan and  Pollefeys uses  spectral clustering on the normalized graph Laplacian matrix of the similarity matrix they propose. However, our approach does not use the spectral clustering on the normalized graph Laplacian of our similarity matrix. Instead, our constructed binary similarity matrix  converts  our original data clustering problem to a simpler clustering of data from 1-dimensional subspaces which can be solved by any traditional data clustering algorithm.

\end{itemize}
\item Our algorithm is reliable in the presence of noise, and applied to the Hopkins 155 Dataset, it generates the best results to  date for motion segmentation.  The two motion, three motion, and overall segmentation rates for the video sequences are 99.43$\%$,  98.69$\%$, and  99.24$\%$, respectively.
\item Many of the subspace segmentation algorithms use SVD to represent the data matrix $W$ as $W=U\Sigma V^t$ and then replace $W$ with the first $r$ rows of $V^t$, where $r$ is the effective rank of $W$. This paper provides a formal justification for this in Proposition \ref{proposition1}.
\end{enumerate}

\subsection{Paper Organization}
The organization of the paper is as follows: Section \ref{preliminaries} gives some preliminaries. In Section \ref{S3}, we devise an algorithm for the subspace segmentation problem in the special case where the  subspaces have equal and known dimensions. In Section \ref{experimentalresults}, we apply our algorithm to the motion segmentation problem,  test it on the Hopkins 155 Datasets, explain the experimental procedure, and present the experimental results.

\section{Preliminaries}
\label{preliminaries}
In this section, we present  Proposition \ref{proposition1} which will be used later to justify that a data matrix $W$ whose columns represent data points can be replaced with a lower rank matrix after computing its SVD (i.e. $W=U\Sigma V^t$). It can be paraphrased by saying that  for any matrices $A,B,C$, a cluster of the columns of $B$ is also a cluster of the columns of $C=AB$. A cluster of $C$ however is not necessarily a cluster  $B$, unless $A$ has full rank:

\begin{proposition}
\label{proposition1}
Let $A$ and $B$ be ${m\times n}$ and $n\times k$ matrices. Let $C=AB$. Assume $J \subset \left\{1,2,\dotsb,k\right\}$.

\begin {enumerate}
 \item If $b_i\in \text{span}\left\{b_j : j\in J\right\}$ then $c_i\in \text{span}\left\{c_j : j\in J\right\}$.
 \item If $A$ is full rank and $m\geq n$ then
 $b_i\in \text{span}\left\{b_j : j\in J\right\}  \Longleftrightarrow c_i\in \text{span}\left\{c_j : j\in J\right\}$

\end {enumerate}    
\label{proposition1}
\end{proposition}
\begin{proof} The first part can be proved by the simple matrix manipulation
\begin{align}
AB &= A\left[ \begin{matrix} b_1  & \dotsb & b_i & \dotsb& b_k \end{matrix} \right] \nonumber\\
&=\left[ \begin{matrix} Ab_1  & \dotsb & Ab_i & \dotsb& Ab_k \end{matrix} \right] \nonumber \\
&=\left[ \begin{matrix} Ab_1  & \dotsb & A\sum_{j\in J}k_jb_j & \dotsb& Ab_k \end{matrix} \right] \nonumber \\
&=\left[ \begin{matrix} Ab_1  & \dotsb & \sum_{j\in J}k_jAb_j & \dotsb& Ab_k \end{matrix} \right] \nonumber \\
&=\left[ \begin{matrix} c_1  & \dotsb & \sum_{j\in J}k_jc_j & \dotsb& c_k \end{matrix} \right]
\end{align}

For the second part, we  note that $A^tA$ is invertible and $(A^tA)^{-1}A^tC=B$. We then apply part 1 of the proposition. Note that the same result clearly holds if $A$ is invertible.
\end{proof}
The proposition above suggest that--for the purpose of column clustering--we can replace a matrix $C$ by matrix $B$ as long as $A$ has the stated properties. Thus by choosing $A$ appropriately the matrix $C$ can be replaced by a more suitable matrix $B$, e.g. $B$ has fewer rows, is better conditioned or is in a format where columns can be easily clustered.

%

\section{Nearness to Local Subspace Approach}
\label {S3}
In this section, we develop a specialized algorithm for subspace segmentation and data clustering when the dimensions of the subspaces are equal and known. First, a local subspace is estimated for each data point. Then, the distances between the local subpaces and points are computed and a distance matrix is generated.  This is followed by construction of a binary similarity matrix  by applying a data-driven threshold to the distance matrix. Finally, the segmentation problem is converted to a one-dimensional data clustering problem. The precise steps are described in Algorithm \ref {algo:greatcircle} and in the explanation that follows.

\subsection{Algorithm for Subspace Segmentation for Subspaces of Equal and Known Dimensions}
The algorithm for subspace segmentation  is given in Algorithm \ref{algo:greatcircle}. We assume that the subspaces have dimension $d$ (for motion segmentation, $d=4$). The details of the various steps are:
\begin{algorithm}
\caption{Subspace Segmentation}
\label{algo:greatcircle}
\begin{algorithmic}[1]
\REQUIRE The $m\times N$ data matrix $W$ whose columns are drawn from subspaces of dimension $d$
\ENSURE Clustering of the feature points.
\STATE Compute the SVD of $W$ as in Equation \eqref{eq:svd}.
	\STATE Estimate the rank of $W$ (denoted by $r$) if it is not known. For example, using Equation \eqref{eq:rankestimation} or any other appropriate choice.
	\STATE Compute $(V_r)^t$ consisting of the first $r$ rows of $V^t$.
	\STATE Normalize the columns of $(V_r)^t$.
	\STATE Replace the data matrix $W$ with $(V_r)^t$.
	\STATE Find the angle between the column vectors of $W$ and represent it as a matrix. \COMMENT{i.e., $\arccos(W^t W)$.}
	\STATE Sort the angles and find the closest neighbors of column vector.
	\FORALL{Column vector $x_i$ of $W$}
			\STATE Find the local subspace for the set consisting of $x_i$ and $k$ neighbors (see Equation \eqref{eq:hypercircle}). \COMMENT{Theoretically, $k$ is  at least $d-1$. We can use the least square approximation for the subspace (see the section \textit{Local Subspace Estimation}). Let $A_i$ denote the matrix whose columns form an orthonormal bases for the local subspace associated with $x_i$.}
	\ENDFOR
	\FOR{$i=1$ to N} 
				\FOR{$j=1$ to N}
						\STATE define $H = (d_{ij}) =\left(||x_j-A_{i}^t x_j||_p+||x_i-A_{j}^t x_i||_p\right)/2$
				\ENDFOR  
					\ENDFOR
	\COMMENT {Build the distance matrix}
	\STATE Sort the entries of the $N\times N$ matrix $H$ from smallest to highest values into the vector $h$ and set the threshold $\eta$ to the value of the $T^{th}$ entry of the sorted and normalized vector $h$, where $T$ is such that $\|\chi_{[T,N^2]}-h\|_2$ is minimized, and where $\chi_{[T,N^2]}$ is the characteristic function of the discrete set $[T,N^2]$.
	\STATE Construct a similarity matrix $S$ by setting all entries of $H$ less than threshold $\eta$  to 1 and by setting all other entries  to 0. \COMMENT {Build the binary similarity matrix}
	\STATE Normalize the rows of $S$ using $l_1$-norm.
	\STATE Perform SVD $S^t = U_n \Sigma_n (V_n)^t$.
	\STATE Cluster the columns of $\Sigma_n$$(V_n)^t$ using k-means. $\Sigma_n(V_n)^t$ is the projection on to the span of $U_n$.
\end{algorithmic}
\end{algorithm}

	\emph{Dimensionality Reduction and Normalization:}
Let $W$  be an $m\times N$ data matrix  whose columns are drawn from a union of subspaces of dimensions at most $d$, possibly perturbed by noise. In order to reduce the dimensionality of the problem, we compute the SVD of $W$ 
\begin {equation}
\label{eq:svd}
W=U\Sigma V^t
\end {equation}
where $U=\left[ \begin{matrix} u_1 & u_2 & \dotsb & u_m \end{matrix} \right] $ is an $m\times m$ matrix, $V=\left[ \begin{matrix} v_1&v_2 & \dotsb & v_N \end{matrix} \right]$ is an $N\times N$ matrix, and $\Sigma$ is an $m\times {\color{red} N}$ diagonal matrix with diagonal entries $\sigma_1,\dots, \sigma_l$, where $l=\min \{m,N\}$. 

To estimate the effective rank of $W$,  one can use the modal selection algorithm \cite{Yan06} to estimate the rank $r$ if it is not known:
\begin{equation}
r=\text{argmin}_r\frac{\sigma_{r+1}^2}{\sum_{i=1}^r\sigma_i^2}+\kappa r
\label{eq:rankestimation}
\end{equation}
where $\sigma_j$ is the $j^{th}$ singular value and $\kappa$ is a suitable constant.  Another possible model selection algorithm can be found in \cite{Zappella11}. $U_r\Sigma_r(V_r)^t$ is the best rank-$r$ approximation of $W=U\Sigma V^t$, where $U_r$ refers to a matrix that has the first $r$ columns of $U$ as its columns and $V_r$ refers to the first $r$ rows of $V^t$. In the case of motion segmentation, if there are $k$ independent motions across the frames captured by a moving camera, the rank of $W$ is between $2(k+1)$ and $4(k+1)$. 
 
We can now replace  the data matrix $W$ with the matrix $(V_r)^t$ that consists of the first $r$ rows of $V^t$ (thereby reducing the dimensionality of data). This step is justified by Proposition \ref{proposition1}.  Also, \cite{Vidal05} discusses the segmentation preserving projections and states that the number of subspaces and their dimensions are preserved by random projections, except for a zero measure set of projections. It should also be noted that this step reduces additive noise as well, especially in the case of light-tailed noise, e.g., Gaussian noise. The number of subspaces corresponds to the number of moving objects. Vidal \textit{et al.} \cite{Vidal08} uses an alternative method (power method) for SVD to project incomplete motion data (trajectories) into a 5-dimensional subspace and then applies GPCA and spectral clustering for subspace segmentation.  Dimensionality reduction corresponds to Steps 1, 2, and 3 in Algorithm \ref{algo:greatcircle}.

Another type of data reduction is normalization. Specifically,  the columns of $(V_r)^t$ are normalized to lie on the unit sphere  $\mathbb{S}^{r-1}$. This is because by projecting the subspace on the unit sphere, we effectively reduce the dimensionality of the data by one. Moreover, the normalization gives equal contribution of the data matrix columns  to the description of the subspaces.  Note that the normalization can be done by using $l_p$ norms of the columns of $(V_r)^t$. This normalization procedure  corresponds to Steps 4 and 5 in Algorithm \ref{algo:greatcircle}. \\ \\
\emph{\textbf{Local Subspace Estimation:}}
The data points (i.e., each column vector of $(V_r)^t$) that are close to each other are likely to belong to the same subspace. For this reason, we estimate a local subspace for each data point using its closest neighbors. This can be done in different ways. For example, if the $l_2$-norm is used for normalization, we can find the angles between the points, i.e., we can compute the matrix $\arccos(V_r\times (V_r)^t)$. Then  we can sort the angles and find the closest neighbors of each point. If we use $l_p$-norm for normalization, we can generate a distance matrix $(a_{ij})=(||x_i-x_j||_p)$ and then sort each column of the distance matrix to find the neighbors of each $x_i$, which is the $i^{th}$ column of $(V_r)^t$.

Once the distance matrix between the points is generated,   we can find,  for each point $x_i$, a set of $k+1\ge d$ points $\left\{x_i, x_{i_1},...,x_{i_k}\right\}$ consisting of $x_i$ and its $k$ closest neighbors. Then we generate a d-dimensional subspace that is nearest (in the least square sense) to the data $\left\{x_i, x_{i_1},...,x_{i_k}\right\}$. This is accomplished by using SVD 
\begin{equation}
\label{eq:hypercircle}
X=\left[x_i \; x_{i_1} \; ... \; x_{i_k}\right] = A\Sigma B^t.
\end{equation}
Let $A_i$ denote the matrix of the first $d$ columns of $A$ associated with $x_i$. Then, the  column space $\textsl{C}(A_i)$  is the $d$-dimensional subspace nearest to $\left\{x_i, x_{i_1},...,x_{i_k}\right\}$. 
Local subspace estimation corresponds to Steps 6 to 10 in Algorithm \ref{algo:greatcircle}. \\ \\
\emph{\textbf{Construction of Binary Similarity Matrix:}}
So far,  we have associated a local subspace $S_i$ to each point $x_i$.   Ideally, the points and only those points that belong to the same subspace as $x_i$ should have zero distance from $S_i$. This suggests computing the distance of each point $x_j$ to the local subspace $S_i$ and forming a distance matrix $H$.

The distance matrix $H$ is generated as $H = (d_{ij}) = \left(||x_j-A_{i}^t x_j||_p+||x_i-A_{j}^t x_i||_p\right)/2$.\\ 
A convenient choice of $p$ is 2. Note that as $d_{ij}$ decreases, the probability of having $x_j$ on the same subspace as $x_i$  increases. Moreover, for $p=2$, $||x_j-A_{i}^t x_j||_2$ is the Euclidean distance of $x_j$ to the subspace associated with $x_i$.

Since we are not in the ideal case, a point $x_j$ that belongs to the same subspace as $x_i$ may have non-zero distance to $S_i$. However, this distance is likely to be small compared to the distance between $x_j$ and $S_k$ if $x_j$ and $x_k$ do not belong to the same subspace. This suggests that we compute a threshold that will distinguish between these two cases and transform the distance matrix into a binary matrix in which a zero  in the $(i,j)$ entry means $x_i$ and $x_j$ are likely to  belong to the same subspace, whereas $(i,j)$ entry of one means   $x_i$ and $x_j$ are not  likely to  belong to the same subspace.

To do this, we convert the distance matrix $H=(d_{ij})_{N\times N}$ into a binary similarity matrix $S=(s_{ij})$. This is done by applying a data-driven thresholding as follows:
\begin{enumerate}
	\item Create a vector $h$ that contains  the sorted entries of $H_{N\times N}$ from smallest to highest values. Scale $h$ so that its smallest value is zero and its largest value is one.
	\item Set the threshold $\eta$ to the value of the $T^{th}$ entry of the sorted vector $h$, where $T$ is such that $\|\chi_{[T,N^2]}-h\|_2$ is minimized, and where $\chi_{[T,N^2]}$ is the characteristic function of the discrete set $[T,N^2]$. 
 If the number of points in each subspace are approximately equal, then we would expect  about $\frac{N}{n}$ points in each subspace, and we would expect $\frac{N^2}{n^2}$ small entries (zero entries ideally). However, this may not be the case in general. For this reason, we compute the data-driven threshold $\eta$ that distinguishes the small entries from the large entries.
 		\item Create a similarity matrix $S$ from $H$ such that all entries of $H$ less than the threshold $\eta$ are set to 1 and  the others are set to 0.
\end{enumerate}
The construction of binary similary corresponds to Steps 11 to 17 in Algorithm \ref{algo:greatcircle}. In \cite{Yan06}, Yan and Pollofeys uses chordal distance (as defined in~\cite{Wong67}) between the subspaces $\F(x_i)$ and $\G(x_j)$  as a measure of the distance between points $x_i$ and $x_j$ 
\begin{equation}
\label{distance_pollefey}
d_{c}^2(\F,\G) = \sum_{i=1}^p \sin^2(\theta_i)
\end{equation}  
where $\{\theta_i\}_{i=1}^{p}$ are the principle angles between $p$-dimensional local subspaces $\F$ and $\G$ with $\theta_1\leq\dots\leq \theta_p$.  In this approach, the distance between any pairs of points from $\F$ and $\G$ is the same. We find distances between points and local subspaces and our approach  distinguishes different points from the same subspace. To see this, let $v \in span\{Q_{\F}\}$, $||v||_2 = 1$, where the columns of $Q_{\F}$ form an orthonormal basis for $\F$. Thus $v=Q_{\F}x$ for some $x$ with $||x||_2=1$. Let $Q_{\G}$ form an orthonormal basis for $\G$, then the Euclidian distance from $v$ to $\G$ squared is given by
\begin{align}
\|v-P_{\G}(v)||^2_2 &= \|Q_{\F}x-Q_{\G}Q_{\G}^tQ_{\F}x\|^2_2 \nonumber \\
&= ||x||_2^2-x^tQ_{\F}^tQ_{\G}Q_{\G}^tQ_{\F}x \nonumber \\
&=||x||_2^2-x^tY\Sigma Z^tZ\Sigma^t Y^tx \nonumber \\
&=x^tYY^tx-x^tY\Sigma \Sigma^t Y^tx \nonumber \\
&=x^tYY^tx-x^tY\Sigma^2Y^tx \nonumber \\ 
&=z\left( I-\Sigma^2\right)z \nonumber 
\end{align}
where $Y\Sigma Z^t$ is the SVD for $Q_{\F}^tQ_{\G}$ and $z:=Y^tx$. Thus, using the relation $\cos\theta_i= \sigma_i $ between principle angles and singular values \cite{Golub96}, we get
\begin{align}
d^2(v,\G) &= \sum_{i=1}^{p}z_i^2\sin^2(\theta_i).
\label{distance_us}
\end{align}
Hence, our approach discriminates distances from points in $\F$ to subspace $\G$. We also  have $\sum_{i=1}^{p}z_i^2\sin^2(\theta_i) \leq  \sum_{i=1}^{p}\sin^2(\theta_i)$ and therefore $d_c$  is more sensitive to  noise.

Using Eq.~\ref{distance_us}, we get $0< \sin\theta_1 \leq d \leq \sin\theta_p$.  Assuming a uniform distribution of samples from $\F$ and $\G$, $h$ can be approximated by a function depicted in Figure \ref{fig:threshold}.  The goal is to find the threshold at the jump discontinuity $T$ from $0$ to $\sin\theta_1$.  Our method minimizes the highlighted area. Under this model, a simple computation shows that our data driven thresholding algorithm picks $T_d=T$ for $\sin\theta_1 
/\sin\theta_p\geq 1/2$, e.g., if $\theta_1 \geq 30^o$. In other situations,  our algorithm overshoots in estimating the threshold index depending on $\theta_1$ and $\theta_{p}$.
\begin{figure}[h]
	\centering
		\includegraphics[scale=0.30]{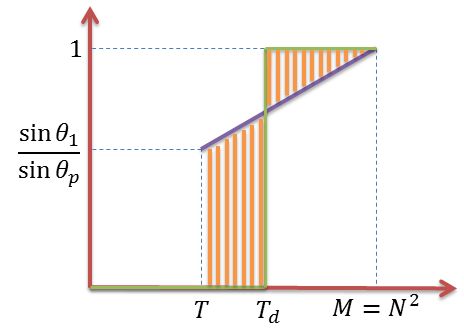}
	\label{fig:threshold}
\caption{Linear modeling for $h$}
\end{figure} \\
\emph{\textbf{Segmentation:}}
The last step is to use the similarity matrix $S$ to segment the data. To do this, we first normalize the rows of $S$ using $l_1$-norm, i.e., $\tilde S = D^{-1}S$, where $D$ is a diagonal matrix $(d_{ij}) = \sum_{j=1}^N s_{ij}$. Note that  $S$ and $\tilde S$ are not symmetric. $\tilde S$ is related to the random walk Laplacian $L_r$ ($\tilde S = I-L_r$) \cite{Petrik07}.  Although other $l_p$ normalizations are possible for $p\geq 1$, however, because of the geometry of the $l_1$ ball, $l_1$-normalization brings outliers closer to the cluster clouds (distances of outliers decrease monotonically as  $p$ decreases to 1).  Since SVD  (which will be used next) is associated with $l_2$ minimization it is sensitive to outliers. Therefore $l_1$ normalization works best when SVD is used. 

Observe that the initial data segmentation problem has now been converted to segmentation of $n$ 1-dimensional subspaces from the rows of $\tilde S$. This is because,
 in the ideal case, from the construction of $\tilde S$,   if $x_i$ and $x_j$ are  in the same subspace, the $i^{th}$ and $j^{th}$ rows of $\tilde S$  are equal. Since there are $n$ subspaces, then there will be $n$ 1-dimensional subspaces.

Now, the problem is again a subspace segmentation problem, but this time the data matrix is $\tilde S$ with each row as a data point. Also, each subspace is 1-dimensional and there are $n$ subspaces. Therefore, we can apply SVD again to obtain
\begin{equation}
\tilde{S}^t = U_n \Sigma_n (V_n)^t. \nonumber
\end{equation}
Using Proposition \ref{proposition1}, it can be shown that $\Sigma_n$$(V_n)^t$ can replace $\tilde{S}^t$ and we cluster the columns of $\Sigma_n$$(V_n)^t$,  which is the projection of $\tilde S$ on to the span of $U_n$. Since the problem is only segmentation of subspaces of dimension 1, we can use any traditional segmentation algorithm such as k-means to cluster the data points. The segmentation corresponds to Steps 18 to 20 in Algorithm \ref{algo:greatcircle}.

\section{Experimental Results}
\label{experimentalresults}
\subsection{The Hopkins 155 Dataset}
The Hopkins 155 Dataset \cite{Vidal07} was created as a benchmark database to evaluate motion segmentation algorithms. It contains two (2) and three (3) motion sequences. There are three (3) groups of video sequences in the dataset: (1) 38 sequences of outdoor traffic scenes captured by a moving camera, (2) 104 indoor checker board sequences captured by a handheld camera, and (3) 13 sequences of articulated motions such as head and face motions. Cornerness features that are extracted and tracked across the frames are  provided along with the dataset.  The ground truth segmentations are also provided for comparison.
%

\subsection{Results}
Tables \ref{tab:twomotion}, \ref{tab:threemotion}, and \ref{tab:overall} display some of the experimental results for the Hopkins 155 Dataset. Our Nearness to Local Subspace (NLS)  approach have been compared with six  (6) motion detection algorithms: (1) GPCA \cite{Vidal05}, (2) RANSAC \cite{Bolles81}, (3) Local Subspace Affinity (LSA) \cite{Yan06}, (4) MLS \cite{Kanatani03,Gruber04},  (5) Agglomerative Lossy Compression (ALC) \cite{Rao101},  and (6) Sparse Subspace Clustering (SSC) \cite{Elhamifar09}. An evaluation of those algorithms is presented in \cite{Elhamifar09} with a minor error in the tabulated results for articulated three motion analysis of SSC-N.  SSC-B and SSC-N correspond to Bernoulli and Normal random projections, respectively \cite{Elhamifar09}. The minor error in \cite{Elhamifar09} is the listing  of error as 1.42\% for articulated three motions. It is replaced with 1.60\% in Table~\ref{tab:threemotion}.  In Tables~\ref{tab:twomotion}-\ref{tab:overall}, we used the number of neighbors $k=3$. Since each point is drawn from a 4-dimensional subspace, a minimum of 3 neighbors are needed to fit a local subspace for each point. Using the same assumption as the algorithms that we compare with,  we take the rank of the data matrix to be 8 for two motion and 12 for three motion.  Table~\ref{tab:twomotion} displays the misclassification rates for the two motions video sequences. NLS outperforms all of the algorithms for the checkerboard sequences, which are linearly independent motions. The overall misclassification rate is 0.57\%. This is 24\% better than the next best algorithm. Table~\ref{tab:threemotion} shows the misclassification rates for  the three motion sequences. NLS has 1.31\% misclassification rate and performs 47\% better than the next best algorithm (i.e. SSC-N). Table~\ref{tab:overall} presents the misclassification rates for all of the video sequences. Our algorithm NLS (with 0.76\% misclassification rate) performs 39\% better than the next best algorithm (i.e. SSC-N). In general, our algorithms outperforms SSC-N, which is given as the best algorithm for the two and three motion sequences together.

Table~\ref{tab:Various_T} shows the performance of the data driven threshold index $T_d$ compared to various other possible thresholds. We  provide the results for $\pm 20\%$, $\pm 10\%$, and $\pm 5\%$ deviations from $T_d$.

Table~\ref{tab:various_k} displays the robustness of the algorithm with respect to the number of neighbors $k$. The second portion of the table excludes one pathological sequence from two-motion checker sequence for $k=4$ and $k=5$. When $k$ is set to 3 - which is the minimum number of neighbors required - the algorithm performs better. 

Table~\ref{tab:LSA} displays the increase in the performance of the original LSA algorithm when our distance/similarity and segmentation techniques are applied separately. Both of them improves the performance of the algorithm, however, the new distance and similarity combination contributes more than the new segmentation technique. 

Recently, the  Low-Rank Representation (LRR) in \cite{Liu10, Liu10_2} was applied to the Hopkins 155 Datasets and it generated an error rate of 3.16\%. The authors state that this error rate can be reduced to 0.87\% by using a variation of LRR with some additional  adjustment of a certain parameter.
\begin{table*}
\centering
\scriptsize
		\begin{tabular}{||c||cccccccc||}
					\hline
				\textbf{\em Checker (78)} &GPCA&LSA&RANSAC&MSL&ALC&SSC-B&SSC-N&NLS\\
					\hline \hline
				Average & 6.09\%    & 2.57\%& 6.52\% &4.46\% & 1.55\%   & 0.83\%  & 1.12\% & 0.23\%\\
				Median & 1.03\%    & 0.27\%& 1.75\% &0.00\% & 0.29\%   & 0.00\%  & 0.00\% & 0.00\% \\
					\hline \hline
				\textbf{\em Traffic (31)}   &GPCA&LSA&RANSAC&MSL&ALC&SSC-B&SSC-N&NLS\\
					\hline \hline
				Average & 1.41\%     & 5.43\%& 2.55\% &2.23\% & 1.59\%   & 0.23\%  & 0.02\% & 1.40\%\\
				Median & 0.00\%    & 1.48\%& 0.21\% &0.00\% & 1.17\%   & 0.00\%  & 0.00\% & 0.00\%\\
					\hline
						\hline \hline
				\textbf{\em Articulated (11)} &GPCA&LSA&RANSAC&MSL&ALC&SSC-B&SSC-N&NLS\\
					\hline \hline
				Average & 2.88\%    & 4.10\%& 7.25\% &7.23\% & 10.70\%   & 1.63\%  & 0.62\% & 1.77\% \\
				Median & 0.00\%   & 1.22\%& 2.64\% &0.00\% & 0.95\%   & 0.00\%  & 0.00\%&  0.88\%\\
				\hline
					\hline \hline
				\textbf{\em All (120 seq)} &GPCA&LSA&RANSAC&MSL&ALC&SSC-B&SSC-N&NLS\\
					\hline \hline
				Average & 4.59\%     & 3.45\%& 5.56\% &4.14\% & 2.40\%   & 0.75\%  & 0.82\%&  \textbf{0.57\%} \\
				Median & 0.38\%   & 0.59\%& 1.18\% &0.00\% & 0.43\%   & 0.00\%  & 0.00\%&  0.00\%\\
				\hline
		\end{tabular}
		\caption{\% segmentation errors for sequences with two motions.}
		\label{tab:twomotion}
\end{table*}
\begin{table*}
\centering
\scriptsize
		\begin{tabular}{||c||cccccccc||}
					\hline
				\textbf{\em Checker (26)} &GPCA&LSA&RANSAC&MSL&ALC&SSC-B&SSC-N&NLS\\
					\hline \hline
				Average & 31.95\%    & 5.80\%& 25.78\% &10.38\% & 5.20\%   & 4.49\%  & 2.97\%&  0.87\% \\
				Median & 32.93\%    & 1.77\%& 26.00\% &4.61\% & 0.67\%   & 0.54\%  & 0.27\%&  0.35\%\\
					\hline \hline
				\textbf{\em Traffic (7)}   &GPCA&LSA&RANSAC&MSL&ALC&SSC-B&SSC-N&NLS\\
					\hline \hline
				Average & 19.83\%  & 25.07\%& 12.83\% &1.80\% & 7.75\%   & 0.61\%  & 0.58\%&  1.86\% \\
				Median & 19.55\%    & 23.79\%& 11.45\% &0.00\% & 0.49\%   & 0.00\%  & 0.00\%& 1.53\% \\
					\hline
						\hline \hline
				\textbf{\em Articulated (2)} &GPCA&LSA&RANSAC&MSL&ALC&SSC-B&SSC-N&NLS\\
					\hline \hline
				Average & 16.85\%    & 7.25\%& 21.38\% &2.71\% & 21.08\%   & 1.60\%  & 1.60\%&  5.12\% \\
				Median & 16.85\%   & 7.25\%& 21.38\% &2.71\% & 21.08\%   & 1.60\%  & 1.60\%&  5.12\% \\
				\hline
					\hline \hline
				\textbf{\em All (35 seq)} &GPCA&LSA&RANSAC&MSL&ALC&SSC-B&SSC-N&NLS\\
					\hline \hline
				Average & 28.66\%     & 9.73\%& 22.94\% &8.23\% & 6.69\%   & 3.55\%  & 2.45\%&  \textbf{1.31\%} \\
				Median & 28.26\%   & 2.33\%& 22.03\% &1.76\% & 0.67\%   & 0.25\%  & 0.20\%&  0.45\% \\
				\hline
		\end{tabular}
		\caption{\% segmentation errors for sequences with three motions.}
		\label{tab:threemotion}
\end{table*}
\begin{table*}
\centering
\scriptsize
		\begin{tabular}{||c||cccccccc||}
					\hline
				\textbf{\em All (155 seq)} &GPCA&LSA&RANSAC&MSL&ALC&SSC-B&SSC-N&NLS\\
					\hline \hline
				Average & 10.34\%     & 4.94\%& 9.76\% &5.03\% & 3.56\%   & 1.45\%  & 1.24\%&  \textbf{0.76\%} \\
				Median & 2.54\%   & 0.90\%& 3.21\% &0.00\% & 0.50\%   & 0.00\%  & 0.00\%&  0.20\%\\
				\hline
		\end{tabular}
		\caption{\% segmentation errors for all sequences.}
		\label{tab:overall}
\end{table*}
\begin{table*}
\centering
\scriptsize
		\begin{tabular}{||c||ccccccc||}
					\hline
				\textbf{\em All-2 (120 seq)}  &Data Driven $T_d$&0.8$T_{d}$&0.9$T_{d}$&0.95$T_{d}$&1.05$T_{d}$&1.10$T_{d}$&1.20$T_{d}$\\
					\hline \hline
				Average  & 0.57\%     & 0.95\% &  1.17\%&  0.62\%& 0.58\%&  1.05\%& 0.77\% \\
				Median & 0.00\%    & 0.00\% &  0.35\%&  2.27\% & 2.27\%&  0.00\%& 0.00\% \\
				\hline
					\hline 
				\textbf{\em All-3 (35 seq)} &Data Driven $T_d$&0.8$T_{d}$&0.9$T_{d}$&0.95$T_{d}$&1.05$T_{d}$&1.10$T_{d}$&1.20$T_{d}$\\
					\hline \hline
				Average & 1.31\%     & 4.39\% &  3.18\% &  1.42\%& 1.20\%&  1.24\%& 2.06\%  \\
				Median& 0.45\%    &  0.60\% &  0.57\% &  0.46\%& 0.45\% &  0.42\%& 0.37\% \\
				\hline \hline
				\textbf{\em All (155 seq)}   &Data Driven $T_d$&0.8$T_{d}$&0.9$T_{d}$&0.95$T_{d}$&1.05$T_{d}$&1.10$T_{d}$&1.20$T_{d}$\\
					\hline 
				Average  & 0.76\%    & 1.84\% &  1.67\% &  0.83\%& 0.74\% &  1.10\%& 1.11\% \\
				Median & 0.20\%     &  0.00\% &  0.00\%&  0.20\%& 0.20\% &  0.18\%& 0.19\% \\
				\hline
		\end{tabular}
		\caption{\% comparison of the data driven threshold index $T_d$ with other choices.}
		\label{tab:Various_T}
\end{table*}
\begin{table*}
\centering
\scriptsize
		\begin{tabular}{||c||cc||c||cc||}
				\multicolumn{1}{c||}{} &
				\multicolumn{2}{c||}{\textit{ALL SEQ INCLUDED}} &
				\multicolumn{1}{c||}{} &
				\multicolumn{2}{c||}{\textit{1 SEQ EXCLUDED}} \\
				\hline 
				\textbf{\em Checker-2 (78)}&k=5&k=4&k=3&k=5&k=4\\
					\hline \hline
				Average    & 0.65\%  & 1.59\% & 0.23\%& 0.23\%  & 0.97\%\\
				Median     & 0.00\%  & 0.00\% & 0.00\% & 0.00\%  & 0.00\%  \\
					\hline \hline
				\textbf{\em Traffic-2 (31)} &k=5&k=4&k=3&k=5&k=4\\
					\hline \hline
				Average    & 1.56\%  & 1.66\% & 1.40\%& 1.56\%  & 1.66\% \\
				Median     & 0.00\%  & 0.00\% & 0.00\%& 0.00\%  & 0.00\% \\
					\hline
						\hline 
				\textbf{\em Articulated-2 (11)} &k=5&k=4&k=3&k=5&k=4\\
					\hline \hline
				Average   & 2.44\%  & 2.33\% & 1.77\%& 2.44\%  & 2.33\%  \\
				Median    & 0.00\%  & 0.00\%&  0.88\%& 0.00\%  & 0.00\%\\
				\hline
					\hline 
				\textbf{\em All-2 (120 seq)} &k=5&k=4&k=3&k=5&k=4\\
					\hline \hline
				Average     & 1.04\%  & 1.75\%&  \textbf{0.57}\%& 0.77\%  & 1.35\% \\
				Median    & 0.00\%  & 0.00\%&  \textbf{0.00\%}& 0.00\%  & 0.00\%\\
				\hline \hline
				\textbf{\em Checker-3 (26)} &k=5&k=4&k=3&k=5&k=4\\
					\hline \hline
				Average     & 0.44\%  & 0.43\%&  0.87\% & 0.44\%  & 0.43\%\\
				Median    & 0.24\%  & 0.22\%&  0.35\%& 0.24\%  & 0.22\%\\
					\hline \hline
				\textbf{\em Traffic-3 (7)}  &k=5&k=4&k=3&k=5&k=4\\
					\hline \hline
				Average     & 6.59\%  & 7.18\%&  1.86\%& 6.59\%  & 7.18\% \\
				Median    & 1.81\%  & 4.37\%& 1.53\%& 1.81\%  & 4.37\% \\
					\hline
						\hline 
				\textbf{\em Articulated-3 (2)} &k=5&k=4&k=3&k=5&k=4\\
					\hline \hline
				Average    & 20.54\%  & 4.05\%&  5.12\%& 20.54\%  & 4.05\% \\
				Median   & 20.54\%  & 4.05\%&  5.12\%& 20.54\%  & 4.05\%\\
				\hline
					\hline 
				\textbf{\em All-3 (35 seq)}&k=5&k=4&k=3&k=5&k=4\\
					\hline \hline
				Average    & 2.82\%  & 1.98\%&  \textbf{1.31\%}& 2.82\%  & 1.98\% \\
				Median  & 0.65\%  & 0.47\%&  \textbf{0.45\%}& 0.65\%  & 0.47\%\\
				\hline \hline
				\textbf{\em All (155 seq)} &k=5&k=4&k=3&k=5&k=4\\
					\hline \hline
				Average     & 1.50\%  & 1.81\%&  \textbf{0.76}\% & 1.30\%  & 1.50\% \\
				Median    & 0.21\%  & 0.00\%&  \textbf{0.20\%}& 0.21\%  & 0.00\%\\
				\hline
		\end{tabular}
		\caption{\% segmentation errors -  NLS algorithm for various $k$.}
		\label{tab:various_k}
\end{table*}
\begin{table*}

\centering
\scriptsize
		\begin{tabular}{||c||ccc||}
					\hline
				\textbf{\em Checker-2 (78)}&LSA(Original)&LSA(New Dist/Similarity)&LSA(New Segmentation)\\
					\hline \hline
				Average & 2.57\%    & 0.97\% & 1.71\%\\
				Median  & 0.27\%    & 0.00\% & 0.00\% \\
					\hline \hline
				\textbf{\em Traffic-2 (31)}&LSA(Original)&LSA(New Dist/Similarity)&LSA(New Segmentation)\\
					\hline \hline
				Average  & 5.43\%     & 1.59\% & 4.99\%\\
				Median  & 1.48\%     & 1.11\% & 0.65\%\\
					\hline
						\hline \hline
				\textbf{\em Articulated-2 (11)} &LSA(Original)&LSA(New Dist/Similarity)&LSA(New Segmentation)\\
					\hline \hline
				Average& 4.10\% & 2.10\% & 4.26\% \\
				Median  & 1.22\%   & 0.43\%&  1.21\%\\
				\hline
					\hline \hline
				\textbf{\em All-2 (120 seq)}  &LSA(Original)&LSA(New Dist/Similarity)&LSA(New Segmentation)\\
					\hline \hline
				Average  & 3.45\%     & 1.22\%&  2.27\% \\
				Median & 0.59\%    & 0.00\%&  0.35\%\\
				\hline
				\textbf{\em Checker-3 (26)}&LSA(Original)&LSA(New Dist/Similarity)&LSA(New Segmentation)\\
					\hline \hline
				Average  & 5.80\%    & 2.66\%&  4.67\% \\
				Median & 1.77\%     & 0.30\%&  0.91\%\\
					\hline \hline
				\textbf{\em Traffic-3 (7)}  &LSA(Original)&LSA(New Dist/Similarity)&LSA(New Segmentation)\\
					\hline \hline
				Average  & 25.07\%    & 6.38\%&  24.46\% \\
				Median & 23.79\%    & 1.28\%& 31.20\% \\
					\hline
						\hline \hline
				\textbf{\em Articulated-3 (2)} &LSA(Original)&LSA(New Dist/Similarity)&LSA(New Segmentation)\\
					\hline \hline
				Average & 7.25\%   & 6.18\%&  7.25\% \\
				Median & 7.25\%    & 6.18\%&  7.25\% \\
				\hline
					\hline \hline
				\textbf{\em All-3 (35 seq)}&LSA(Original)&LSA(New Dist/Similarity)&LSA(New Segmentation)\\
					\hline \hline
				Average & 9.73\%     & 2.45\%&  8.78\% \\
				Median& 2.33\%    & 0.20\%&  1.94\% \\
				\hline
				\textbf{\em All (155 seq)} &LSA(Original)&LSA(New Dist/Similarity)&LSA(New Segmentation)\\
					\hline \hline
				Average  & 4.94\%    & 1.84\%&  3.96\% \\
				Median & 0.90\%     & 0.18\%&  0.61\%\\
				\hline
		\end{tabular}
		\caption{\% segmentation errors for LSA with various parameters.}
		\label{tab:LSA}

\end{table*}

\section{Conclusions}
 The NLS  approach described in this paper can handle noise  effectively, but it works only in  special cases of subspaces segmentation problems (i.e., subspaces of equal and known dimensions).  Our approach is based on the computation of a binary similarity matrix for the data points. A local subspace is first estimated for each data point. Then, a distance matrix is generated by computing the distances between the local subspaces and points. The distance matrix is converted to the similarity matrix by applying a data-driven threshold.  The problem is then transformed to segmentation of subspaces of dimension $1$ instead of subspaces of dimension $d$.   The algorithm was applied to the Hopkins 155 Dataset and generated the best results to  date. 

\section*{Acknowledgement}
We would like to thank Professor Ren\'{e} Vidal for his invaluable comments and feedback.
\bibliographystyle{elsarticle-num}
\bibliography{sekmen}

\end{document}